\documentclass[12pt,a4paper]{article}

\usepackage[utf8]{inputenc}
\usepackage{fullpage}
\usepackage[english]{babel}
\usepackage{amsmath}
\usepackage{amssymb}
\usepackage{amsthm}
\usepackage[hidelinks]{hyperref}
\usepackage{tikz-cd}
\usepackage{tikz-inet}

\tikzcdset{arrow style=tikz}
\newtheorem{proposition}{Proposition}
\theoremstyle{definition}
\newtheorem{definition}{Definition}
\newtheorem*{notation}{Notation}

\newcommand{\interr}{\xrightarrow{+}}
\newcommand{\inter}{\xrightarrow{\bowtie}}
\newcommand{\indir}{\xrightarrow{:=}}
\newcommand{\arity}{\operatorname{ar}}

\title{Upward confluence in the interaction calculus}
\author{Anton Salikhmetov}

\begin{document}
\maketitle

\begin{abstract}
The lambda calculus is not upward confluent, one of counterexamples known thanks to Plotkin.
This paper investigates upward confluence in the interaction calculus.
Can an interaction system have this property?
We positively answer this question and also provide a necessary and sufficient condition for stronger one-step upward confluence which happens to be very restrictive.
However, the provided condition is not necessary for upward confluence as we prove that the interaction system of the linear lambda calculus is upward confluent.
\end{abstract}

\section{Introduction}

In the context of rewriting systems, upward confluence also known as upside down CR and upward Church-Rosser is the following property.

\begin{definition}
\label{rsys}
A rewriting system $R$ is \textit{upward confluent} if and only if
$$
\forall s_1, s_2, s: s_1 \rightarrow^*_R s\ \wedge\ s_2 \rightarrow^*_R s\ \Rightarrow\ \exists s': s' \rightarrow^*_R s_1\ \wedge\ s' \rightarrow^*_R s_2.
$$
\end{definition}

The $\lambda$-calculus has the following counterexample \cite[Exercise 3.5.11 (vii)]{barendregt} thanks to Plotkin: $\lambda$-terms $(\lambda x.b\ x\ (b\ c))\ c$ and $(\lambda x.x\ x)\ (b\ c)$ both $\beta$-reduce to $(b\ c)\ (b\ c)$ without any common $\beta$-expand; see also \cite{vanoostrom} where this counterexample and another one are discussed.

It is also useful to distinguish a one-step version of upward confluence.

\begin{definition}
\label{srsys}
A rewriting system $R$ is \textit{strongly upward confluent} if and only if
$$
\forall s_1, s_2, s: s_1 \rightarrow_R s\ \wedge\ s_2 \rightarrow_R s\ \Rightarrow\ \exists s': s' \rightarrow_R s_1\ \wedge\ s' \rightarrow_R s_2.
$$
\end{definition}

Definition~\ref{srsys} demands what can also be called the \textit{upward diamond property}.
Note that the upward diamond property implies upward confluence, but not vice versa.

\section{Irrelevance of indirection}

Let us investigate how Definition~\ref{rsys} and Definition~\ref{srsys} apply to the interaction calculus~\cite{calculus}.

\begin{notation}
Let $c \rightarrow c'$, where $c$ and $c'$ are configurations.
We write $c \inter c'$ when $c \rightarrow c'$ is interaction, and $c \indir c'$ when $c \rightarrow c'$ is indirection, where indirection will mean any substitution, including that in the interface of a configuration.
\end{notation}

\begin{proposition}
\label{indupdiamond}
$\forall c_1, c_2, c: c_1 \indir c\ \wedge\ c_2 \indir c\ \Rightarrow\ \exists c' : c' \indir c_1\ \wedge\ c' \indir c_2$.
\end{proposition}
\begin{proof}
If $c_1 \neq c_2$, then indirections $c_1 \indir c$ and $c_2 \indir c$ remove two different names $x$ and $y$, respectively, and we have one of the following three possible cases:
\begin{enumerate}
\item $c = \langle \dots v[x := t] \dots w[y := u]\dots\ \Delta \rangle$, $x$ occurs in $v$, and $y$ occurs in $w$. \\
Then $c_1 = \langle \dots v \dots w[y := u]\dots\ x = t,\ \Delta \rangle$ and \\
\phantom{Then} $c_2 = \langle \dots v[x := t] \dots w \dots\ y = u,\ \Delta \rangle$. \\
Choose $c' = \langle \dots v \dots w \dots\ x = t,\ y = u,\ \Delta \rangle$.

\item $c = \langle \dots v[x := t][y := u] \dots\ \Delta \rangle$, and both $x$ and $y$ occur in $v$. \\
Notice $v[x := t][y := u] \equiv v[y := u][x := t]$. \\
Then $c_1 = \langle \dots v[y := u] \dots\ x = t,\ \Delta \rangle$ and \\
\phantom{Then} $c_2 = \langle \dots v[x:=t]\dots\ y = u,\ \Delta \rangle$. \\
Choose $c' = \langle \dots v \dots\ x = t,\ y = u,\ \Delta \rangle$.

\item $c = \langle \dots v[x := t[y := u]] \dots\ \Delta \rangle$, $x$ occurs in $v$, and $y$ occurs in $t$. \\
Notice $v[x := t[y := u]] \equiv v[x := t][y := u]$. \\
Then $c_1 = \langle \dots v \dots\ x = t[y := u],\ \Delta \rangle$ and \\
\phantom{Then} $c_2 = \langle \dots v[x:=t] \dots\ y = u,\ \Delta \rangle$. \\
Choose $c' = \langle \dots v \dots\ x = t,\ y = u,\ \Delta \rangle$.
\end{enumerate}
If $c_1 = c_2 = \langle \vec f\ |\ x = t,\ \Delta \rangle$, choose $c' = \langle \vec f\ |\ x = y,\ y = t,\ \Delta \rangle$ with a fresh name $y$.
\end{proof}

\begin{notation}
If $\alpha[v_1,\dots,v_m] \bowtie \beta[w_1,\dots,w_n]$, then $\alpha(t_1,\dots,t_m) \div \beta(u_1,\dots,u_n)$ denotes the following multiset of equations: $\{t_1=v_1,\dots,t_m=v_m,\ u_1=w_1,\dots,u_n=w_n\}$.
\end{notation}

\begin{proposition}
\label{indintcom}
$\forall c_1, c_2, c: c_1 \inter c\ \wedge\ c_2 \indir c\ \Rightarrow\ \exists c' : c' \indir c_1\ \wedge\ c' \inter c_2$.
\end{proposition}
\begin{proof}
Assume $c_1 \inter c'$ due to $\alpha[\vec v] \bowtie \beta[\vec w]$ and $c_2 \indir c'$ removes a name $x$.
Interaction cannot remove any name, hence $c_1 \neq c_2$.
Now we have the following two possible cases:
\begin{enumerate}
\item $c = \langle \dots t'[x:=u'] \dots\ \alpha(\vec t)\div\beta(\vec u),\ \Delta \rangle$, and $x$ occurs in $t'$. \\
Then $c_1 = \langle \dots t'[x:=u'] \dots\ \alpha(\vec t)=\beta(\vec u),\ \Delta \rangle$ and \\
\phantom{Then} $c_2 = \langle \dots t' \dots\ x = u',\ \alpha(\vec t)\div\beta(\vec u),\ \Delta \rangle$. \\
Choose $c' = \langle \dots t' \dots\ x = u',\ \alpha(\vec t)=\beta(\vec u),\ \Delta \rangle$.

\item $c = \langle \vec f\ |\ \alpha(\vec t\,[x:=t'])\div\beta(\vec u),\ \Delta \rangle$, and $x$ occurs in $\vec t$. \\
Then $c_1 = \langle \vec f\ |\ \alpha(\vec t\,[x:=t'])=\beta(\vec u),\ \Delta \rangle$ and \\
\phantom{Then} $c_2 = \langle \vec f\ |\ x = t',\ \alpha(\vec t)\div\beta(\vec u),\ \Delta \rangle$. \\
Choose $c' = \langle x = t',\ \alpha(t_1,\dots,t_m)=\beta(u_1,\dots,u_n),\ \Delta \rangle$.
\end{enumerate}
In each case, $c' \indir c_1$ and $c' \inter c_2$, so interaction and indirection commute upwards.
\end{proof}

Propositions \ref{indupdiamond} and \ref{indintcom} show that upward confluence is determined by the interaction rules and has little to do with indirections.
The latter makes sense because indirections in the interaction calculus are rather an artifact of textual representation of interaction nets that is irrelevant in the original graphical representation.
As an immediate consequence, interaction systems without any interaction rules always have strong upward confluence.

\begin{notation}
We write $c \interr c'$ when $c \inter c_1 \indir \cdots \indir c_n = c'$ for some $c_i$ and $n \ge 1$.
\end{notation}

\begin{proposition}
\label{plusupdiamond}
An interaction system is upward confluent if
\begin{equation}
\label{plusdiamond}
\forall c_1, c_2, c: c_1 \inter c\ \wedge\ c_2 \inter c\ \Rightarrow\ \exists c': c' \interr c_1\ \wedge\ c' \interr c_2.
\end{equation}
\end{proposition}
\begin{proof}
First, let us demonstrate that (\ref{plusdiamond}) implies the upward diamond property for $\interr$ by using Proposition~\ref{indupdiamond}, Proposition~\ref{indintcom}, and the diagram in Figure~\ref{diamondfig}.
\begin{figure}
$$
\begin{tikzcd}
& & c
\ar[rrdd, out=0, in=90, thick, dashed, "+"]
\ar[lldd, out=180, in=90, thick, dashed, "+"']
\ar[ld, dashed, "+"]
\ar[rd, dashed, "+"']
& & \\
& c_1
\ar[ld, two heads, dashed, ":="]
\ar[rd, dashed, "\bowtie"]
& & c_2
\ar[ld, dashed, "\bowtie"']
\ar[rd, dashed, two heads, ":="'] & \\
c'_1
\ar[rrdd, thick, "+"', out=-90, in=180]
\ar[rd, "\bowtie"]
& & c'
\ar[ld, dashed, two heads, ":="]
\ar[rd, dashed, two heads, ":="']
& & c'_2
\ar[ld, "\bowtie"']
\ar[lldd, out=-90, in=0, thick, "+"] \\
& c''_1 \ar[rd, two heads, ":="] &
& c''_2 \ar[ld, two heads, ":="'] & \\
& & c'' & &
\end{tikzcd}
$$
\caption{$\forall c'_1, c'_2, c'': c'_1 \interr c''\ \wedge\ c'_2 \interr c''\ \Rightarrow\ \exists c: c \interr c'_1\ \wedge\ c \interr c'_2$.}
\label{diamondfig}
\end{figure}
Further, since indirection has the upward diamond property, so does $\indir\cup\interr$.
Now, notice that $\rightarrow^*$ is the transitive closure of $\indir\cup\interr$.
Therefore $\rightarrow^*$ also has the upward diamond property.
\end{proof}

Note that $\interr$ is indistinguishable from reduction in interaction nets, thus Proposition~\ref{plusupdiamond} shows that strong upward confluence of an interaction system in interaction nets implies upward confluence of the corresponding interaction system in the interaction calculus.

\section{Reversible interaction systems}
\label{secris}

Now let us find necessary and sufficient conditions for strong upward confluence.

\begin{definition}
Two active pairs $\alpha(\vec t)=\beta(\vec u)$ and $\alpha'(\vec{t'})=\beta'(\vec{u'})$ \textit{clash} if and only if
$$
\alpha(\vec t)\div\beta(\vec u) = \alpha'(\vec{t'})\div\beta'(\vec{u'}) \neq \varnothing.
$$
\end{definition}

For instance, active pairs $\gamma(t, u) = \gamma(v, w)$ and $\delta(t, u) = \delta(w, v)$ clash in the system of interaction combinators, but $\varepsilon = \varepsilon$ does not clash with any active pair, even with itself.

\begin{proposition}
\label{noclash}
If two different active pairs clash in an interaction system, then the interaction system is not strongly upward confluent.
\end{proposition}
\begin{proof}
Let two different active pairs $\alpha(\vec t)=\beta(\vec u)$ and $\alpha'(\vec{t'})=\beta'(\vec{u'})$ clash in a strongly upward confluent interaction system.
Then there exist the following configurations:
\begin{align*}
c_1 &= \langle \vec f\ |\ \alpha(\vec t)=\beta(\vec u)\rangle, \\
c_2 &= \langle \vec f\ |\ \alpha'(\vec{t'})=\beta'(\vec{u'})\rangle \neq c_1,\ \text{and} \\
c' &= \langle \vec f\ |\ \alpha(\vec t)\div\beta(\vec u)\rangle = \langle \vec f\ |\ \alpha'(\vec{t'})\div\beta'(\vec{u'})\rangle.
\end{align*}
Since $c_1 \inter c'$ and $c_2 \inter c'$, we also have $c \rightarrow c_1$ and $c \rightarrow c_2$ for some $c$.
Interaction cannot remove any name, so if $c \indir c_1$ or $c \indir c_2$, then we would have $c \indir c_1 \Leftrightarrow c \indir c_2$ and $c_1 = c_2$.
Thus $c \inter c_1$ and $c \inter c_2$.
Further, $c_1 \neq c_2$ implies that $c$ has at least two active pairs, namely $\alpha(\vec t)=\beta(\vec u)$ and $\alpha'(\vec{t'})=\beta'(\vec{u'})$.
So $c = \langle \vec f\ |\ \alpha(\vec t)=\beta(\vec u),\ \alpha'(\vec{t'})=\beta'(\vec{u'}),\ \Delta\rangle$ for some $\Delta$.
However, $\alpha(\vec t)\div\beta(\vec u) \rangle = \alpha'(\vec{t'})\div\beta'(\vec{u'}) \neq \varnothing$ results in a contradiction.
\end{proof}

\begin{definition}
An interaction rule $\alpha[\vec v]\bowtie\beta[\vec w]$ is \textit{disconnected} if and only if $\alpha(\vec x)\div\beta(\vec y)$, where $\vec x$ and $\vec y$ are fresh names, can be split into some non-empty multisets $\Gamma$ and $\Delta$ such that no name occurs in both $\Gamma$ and $\Delta$; otherwise $\alpha[\vec v]\bowtie\beta[\vec w]$ is \textit{connected}.
\end{definition}

For instance, annihilation rules $\gamma[x, y] \bowtie \gamma[y, x]$ and $\delta[x, y] \bowtie \delta[x, y]$ are disconnected, whereas duplication $\gamma[\delta(x_1, x_2), \delta(y_1, y_2)] \bowtie \delta[\gamma(x_1, y_1), \gamma(x_2, y_2)]$ is connected.

\begin{proposition}
\label{connected}
If an interaction system is strongly upward confluent, then each of its interaction rules is connected.
\end{proposition}
\begin{proof}
Let $\alpha[\vec v]\bowtie\beta[\vec w]$ be disconnected in a strongly upward confluent interaction system.
Take fresh names $\vec x$ and $\vec y$ and split $\alpha(\vec x)\div\beta(\vec y)$ into two non-empty multisets $\Gamma$ and $\Delta$ such that no name occurs in both of them.
Denote the names from $\vec x$ and $\vec y$ that occur in $\Gamma$ as $\vec{x'}$ and the names from $\vec x$ and $\vec y$ that occur in $\Delta$ as $\vec{y'}$.
Now choose
\begin{align*}
c_1 &= \langle\vec{x'}, \vec{y'}, \vec{y''}\ |\ \alpha(\vec x\,[\vec{y'}:=\vec{y''}])=\beta(\vec y\,[\vec{y'}:=\vec{y''}]),\ \Delta\rangle, \\
c_2 &= \langle\vec{x'}, \vec{y'}, \vec{y''}\ |\ \alpha(\vec x)=\beta(\vec y),\ \Delta[\vec{y'}:=\vec{y''}] \rangle \neq c_1,\ \text{and} \\
c' &= \langle\vec{x'}, \vec{y'}, \vec{y''}\ |\ \Gamma,\ \Delta,\ \Delta[\vec{y'}:=\vec{y''}]\rangle.
\end{align*}
Notice $c_1 \inter c'$ and $c_2 \inter c'$.
Similarly to Proposition~\ref{noclash}, we have $c \inter c_1$ and $c \inter c_2$ for some $c$.
Since $c_1 \neq c_2$, configuration $c$ has at least two active pairs, namely $\alpha(\vec x)=\beta(\vec y)$ and $\alpha(\vec x\,[\vec{y'}:=\vec{y''}])=\beta(\vec y\,[\vec{y'}:=\vec{y''}])$.
However, $\Gamma \neq \varnothing$ results in a contradiction.
\end{proof}

Propositions \ref{noclash} and \ref{connected} together motivate the following definition that combines both found necessary conditions for strong upward confluence of an interaction system.

\begin{definition}
\label{revdef}
An interaction system is \textit{reversible} if and only if no two different active pairs clash and each rule is connected; otherwise the interaction system is \textit{irreversible}.
\end{definition}

In interaction nets, a reversible interaction system is such that two active pairs never reduce to overlapping nets and any net can be the result of at most one active pair.

\begin{proposition}
\label{unambsc}
Any reversible interaction system is strongly upward confluent.
\end{proposition}
\begin{proof}
Let $c_1 \rightarrow c$ and $c_2 \rightarrow c$.
If $c_1 \indir c$ or $c_2 \indir c$, use Proposition~\ref{indupdiamond} or Proposition~\ref{indintcom}.
Otherwise, assume that $c_1 \inter c$ due to $\alpha[\vec v]\bowtie\beta[\vec w]$ and $c_2 \inter c$ due to (possibly the same rule) $\alpha'[\vec{v'}]\bowtie\beta'[\vec{w'}]$.
If $c_1 \neq c_2$, then $c_1$ and $c_2$ can be written as $\langle \vec f\ |\ \alpha(\vec t) = \beta(\vec u),\ \Gamma_1,\ \Delta \rangle$ and $\langle \vec f\ |\ \alpha'(\vec{t'}) = \beta'(\vec{u'}),\ \Gamma_2,\ \Delta \rangle$, respectively, so that $\Gamma_1 \cap \Gamma_2 = \varnothing$.
Therefore we have
$$
c = \langle \vec f\ |\ \alpha(\vec t) \div \beta(\vec u),\ \Gamma_1,\ \Delta \rangle = \langle \vec f\ |\ \alpha'(\vec{t'}) \div \beta'(\vec{u'}),\ \Gamma_2,\ \Delta \rangle.
$$
Notice $\Gamma_1 \subseteq \alpha'(\vec{t'}) \div \beta'(\vec{u'})$ and $\Gamma_2 \subseteq \alpha(\vec t) \div \beta(\vec u)$.
Now the following cases are exhaustive:
\begin{enumerate}
\item $\Gamma_1 = \alpha'(\vec{t'}) \div \beta'(\vec{u'})$ or $\Gamma_2 = \alpha(\vec t) \div \beta(\vec u)$. \\
Notice $\Gamma_1 = \alpha'(\vec{t'}) \div \beta'(\vec{u'}) \Leftrightarrow \Gamma_2 = \alpha(\vec t) \div \beta(\vec u)$. \\
Choose $c' = \langle \vec f\ |\ \alpha(\vec t) = \beta(\vec u),\ \alpha'(\vec{t'}) = \beta'(\vec{u'}),\ \Delta\rangle$.

\item $\Gamma_1 \subsetneq \alpha'(\vec{t'}) \div \beta'(\vec{u'})$ and $\Gamma_2 \subsetneq \alpha(\vec t) \div \beta(\vec u)$. \\
Notice that $\alpha'(\vec{t'}) \div \beta'(\vec{u'}) \neq \varnothing$ and $\alpha(\vec t) \div \beta(\vec u) \neq \varnothing$. \\
If $\Gamma_1 \neq \varnothing$ or $\Gamma_2 \neq \varnothing$, then either of $\alpha'[\vec{v'}]\bowtie\beta'[\vec{w'}]$ and $\alpha[\vec v]\bowtie\beta[\vec w]$ is disconnected. \\
Otherwise $\Gamma_1 = \Gamma_2 = \varnothing$, which makes $\alpha(\vec t)=\beta(\vec u)$ and $\alpha'(\vec{t'})=\beta'(\vec{u'})$ clash.
\end{enumerate}
If $c_1 = c_2 = \langle \vec f\ |\ t = u,\ \Delta \rangle$, choose $c' = \langle \vec f\ |\ t = x,\ x = u,\ \Delta \rangle$ with a fresh name $x$.
\end{proof}

Propositions \ref{noclash}, \ref{connected}, and \ref{unambsc} together show that reversibility in the sense of Definition~\ref{revdef} is a necessary and sufficient condition for strong upward confluence of an interaction system.

\section{Towards the reversible combinators}
\label{revcomb}

Let us have a closer look at the properties of reversible interaction systems.

\begin{definition}
The \textit{arity} of an interaction rule $\alpha[\vec v]\bowtie\beta[\vec w]$ is $\arity(\alpha)+\arity(\beta)$.
\end{definition}

\begin{proposition}
\label{diffarity}
No two different rules have the same positive arity in a reversible system.
\end{proposition}
\begin{proof}
Assume that a reversible interaction system has the following two rules:
$$
\alpha_1[v_1,\dots,v_m]\bowtie\beta_1[v_{m+1},\dots,v_{m+n}]\ \text{and}\ \alpha_2[w_1,\dots,w_{m'}]\bowtie\beta_2[w_{m'+1},\dots,w_{m'+n'}],
$$
where $m + n = m' + n' \ge 1$.
Then we have a clash between two different active pairs:
$$
\alpha_1(w_1,\dots,w_m)\div\beta_1(w_{m+1},\dots,w_{m+n})=\alpha_2(v_1,\dots,v_{m'})\div\beta_2(v_{m'+1},\dots,v_{m'+n'})\neq\varnothing.
$$
The latter contradicts reversibility of the interaction system.
\end{proof}

\begin{definition}
\label{revrule}
An interaction rule $\alpha[\vec v]\bowtie\beta[\vec w]$ is \textit{reversible} if and only if the rule is connected and no two different active pairs $\alpha(\vec t)=\beta(\vec u)$ and $\alpha(\vec{t'})=\beta(\vec{u'})$ clash.
\end{definition}

\begin{proposition}
An interaction system is reversible if and only if all its rules are reversible and no two different rules in the interaction system have the same positive arity.
\end{proposition}
\begin{proof}
($\Rightarrow$) follows from Definition~\ref{revdef}, Definition~\ref{revrule}, and Proposition~\ref{diffarity}.
Conversely, ($\Leftarrow$) is due to Definition~\ref{revdef}, Definition~\ref{revrule}, and the fact that any clash requires the same arity.
\end{proof}

One may wonder if there exists a reversible equivalent of interaction combinators~\cite{comb}.
In other words, is there a reversible interaction system into which any other reversible interaction system can be translated in the sense of \cite{comb}? Assuming that such a system exists, we will refer to it as the \textit{reversible combinators}.

\begin{definition}
An interaction system is \textit{trivial} if and only if $\forall \alpha \in \Sigma: \arity(\alpha) = 0$.
\end{definition}

It is easy to see that the system of reversible combinators cannot be trivial.
Similarly to the system of interaction combinators, one may look for a complete interaction system that has an interaction rule for each pair $(\alpha, \beta) \in \Sigma^2$ with a finite signature $\Sigma$.

\begin{proposition}
If an interaction system is reversible, complete, and non-trivial, then
$$
\forall \alpha, \beta \in \Sigma: \arity(\alpha) = \arity(\beta) \Rightarrow \alpha = \beta.
$$
\end{proposition}
\begin{proof}
As the interaction system is non-trivial, we have $\arity(\alpha) \neq 0$ for some $\alpha \in \Sigma$.
Assume $\beta_1 \neq \beta_2$ and $\arity(\beta_1) = \arity(\beta_2) = n$.
If $n = 0$, then $\beta_1 \neq \alpha \neq \beta_2$ and the rules for $(\alpha, \beta_1)$ and $(\alpha, \beta_2)$ have the same positive arity.
Otherwise, $n > 0$ and the rules for $(\beta_1, \beta_2)$ and $(\beta_1, \beta_1)$ have the same positive arity.
Both contradict Proposition~\ref{diffarity}.
\end{proof}

\section{Upward confluence without reversibility}
\label{seclin}

Although only a reversible interaction system can be \textit{strongly} upward confluent, there still exist irreversible systems that are upward confluent.
In particular, consider the interaction system of the linear $\lambda$-calculus whose only rule is $@[x, y]\bowtie\lambda[x, y]$.
Not only the interaction rule is disconnected, but we also have a clash:
$$
@(t, u)\div\lambda(v, w) = @(v, u)\div\lambda(t, w) \neq \varnothing.
$$

\begin{proposition}
The interaction system of the linear $\lambda$-calculus is upward confluent. 
\end{proposition}
\begin{proof}
Identical to Proposition~\ref{unambsc} except when $\Gamma_1 \subsetneq @(\vec{t'}) \div \lambda(\vec{u'})$ and $\Gamma_2 \subsetneq @(\vec t) \div \lambda(\vec u)$.
Let $\alpha_1, \alpha_2, \beta_1, \beta_2 \in \{@, \lambda\}$, $\alpha_1 \neq \beta_1$, and $\alpha_2 \neq \beta_2$.
Now we have the following two cases:
\begin{enumerate}
\item $\Gamma_1 \neq \varnothing$ or $\Gamma_2 \neq \varnothing$.
Notice $\Gamma_1 \neq \varnothing \Leftrightarrow \Gamma_2 \neq \varnothing$.
Then for some $t_i$ and $u_i$ we have \\
$c = \langle \vec f\ |\ t_1 = x,\ x = u_1,\ t_2 = y,\ y = u_2,\ t_3 = z,\ z = u_3,\ \Delta\rangle$. \\
If $c_1 = \langle \vec f\ |\ \alpha_1(t_1, t_2) = \beta_1(u_1, u_2),\ t_3 = z,\ z = u_3,\ \Delta\rangle$ and \\
\phantom{If} $c_2 = \langle \vec f\ |\ \alpha_2(t_2, t_3) = \beta_2(u_2, u_3),\ t_1 = z,\ z = u_1,\ \Delta\rangle$, then choose \\
\phantom{If} $c' = \langle \vec f\ |\ \alpha_1(t_1, t_2) = \beta_1(u_1, x),\ \alpha_2(x, t_3) = \beta_2(u_2,u_3),\ \Delta\rangle$ shown in Figure~\ref{disfig}. \\
Similarly for the other possible $c_1$ and $c_2$.

\item $\Gamma_1 = \Gamma_2 = \varnothing$.
Then for some $t_i$ and $u_i$ we have \\
$c = \langle \vec f\ |\ t_1 = x,\ x = u_1,\ t_2 = y,\ y = u_2,\ \Delta\rangle$. \\
If $c_1 = \langle \vec f\ |\ \alpha_1(t_1, t_2) = \beta_1(u_1, u_2),\ \Delta\rangle$ and \\
\phantom{If} $c_2 = \langle \vec f\ |\ \alpha_2(t_1, t_2) = \beta_2(u_1, u_2),\ \Delta\rangle$, then choose \\
\phantom{If} $c' = \langle \vec f\ |\ \alpha_1(t_1, t_2) = \beta_1(x, y),\ \alpha_2(x, y) = \beta_2(u_1, u_2),\ \Delta\rangle$ shown in Figure~\ref{ambfig}. \\
Similarly for the other possible $c_1$ and $c_2$.
\end{enumerate}
For each chosen configuration $c'$, notice $c' \interr c_1$ and $c' \interr c_2$, then use Proposition~\ref{plusupdiamond}.
\end{proof}

\begin{figure}
\begin{minipage}{0.5\textwidth}
$$
\begin{tikzpicture}[baseline=(a1.base)]
\matrix[row sep=1.2em]{
& \node[circle, draw] (a1) {$\alpha_1$}; &
& \node[circle, draw] (a2) {$\alpha_2$}; &
& &
& \node[circle, draw] (b1) {$\beta_1$}; &
& \node[circle, draw] (b2) {$\beta_2$}; & \\
\node (t1) {$t_1$}; & &
\node (t2) {$t_2$}; & &
\node (t3) {$t_3$}; &
\node (t4) {}; &
\node (u1) {$u_1$}; & &
\node (u2) {$u_2$}; & &
\node (u3) {$u_3$}; \\
};
\draw[>-<, >=latex] (a1) -- ++(0,0.9) -| (b1);
\draw[>-<, >=latex] (a2) -- ++(0,1.2) -| (b2);
\draw (t2) -- (a1) -- (t1);
\draw (u2) -- (b2) -- (u3);
\draw (t3) -- (a2);
\draw (u1) -- (b1);
\draw (a2) -- ++(0,-0.7) -| (b1);
\end{tikzpicture}
$$
\caption{$\Gamma_1 \neq \varnothing$ and $\Gamma_2 \neq \varnothing$.}
\label{disfig}
\end{minipage}%
\begin{minipage}{0.5\textwidth}
$$
\begin{tikzpicture}[baseline=(a1.base)]
\matrix[row sep=1.2em]{
& \node[circle, draw] (a1) {$\alpha_1$}; &
& \node[circle, draw] (a2) {$\alpha_2$}; &
& &
& \node[circle, draw] (b1) {$\beta_1$}; &
& \node[circle, draw] (b2) {$\beta_2$}; & \\
\node (t1) {$t_1$}; & &
\node (t2) {$t_2$}; & &
\node (t3) {}; &
\node (t4) {}; &
\node (u1) {}; & &
\node (u2) {$u_1$}; & &
\node (u3) {$u_2$}; \\
};
\draw[>-<, >=latex] (a1) -- ++(0,0.9) -| (b1);
\draw[>-<, >=latex] (a2) -- ++(0,1.2) -| (b2);
\draw (t2) -- (a1) -- (t1);
\draw (u2) -- (b2) -- (u3);
\draw (a2) -- ++(-0.3,-0.7) -| (b1);
\draw (b1) -- ++(0.4,-0.9) -| (a2);
\end{tikzpicture}
$$
\caption{$\Gamma_1 = \Gamma_2 = \varnothing$.}
\label{ambfig}
\end{minipage}
\end{figure}

\section{Conclusion}

The current paper has touched upon the aspect of upward confluence in the context of interaction nets.
In Section~\ref{secris}, we introduced the notion of a reversible interaction system and showed that reversibility is a necessary and sufficient condition for strong upward confluence.
Still, upward confluence is proved to be possible without reversibility using the linear $\lambda$-calculus as a rather curious example in Section~\ref{seclin}.

We believe that reversible interaction systems are worth further study not only because of their strong upward confluence but also with respect to reversible computation in a local sense that disregards the order in which independent computations take place.

One possible way to further study reversible systems is to find at least one universal reversible interaction system.
As discussed in Section~\ref{revcomb}, such a system should be complete and have a finite signature with unique arities for all its agent types.


\begin{thebibliography}{0}
\bibitem{barendregt} Henk P. Barendregt, 1984. \\
\textit{The Lambda Calculus: Its Syntax and Semantics.} \\ Studies in Logic and the Foundations of Mathematics, 103.

\bibitem{vanoostrom} Vincent van Oostrom, 1996. \\
\textit{An Easy Expansion Exercise.} \\
{\small \url{https://semanticscholar.org/paper/c179a23274c06aaca4865a2e3c66a45612c38aa0}}

\bibitem{calculus} Maribel Fern\'andez and Ian Mackie, 1999. \\
\textit{A calculus for interaction nets.} \\
Principles and Practice of Declarative Programming, pp.~170--187.

\bibitem{comb} Yves Lafont, 1997. \\
\textit{Interaction combinators.} \\
Information and Computation, 137(1), pp.~69--101.
\end{thebibliography}
\end{document}